\documentclass[12pt]{article}

\usepackage[margin=1in]{geometry}
\usepackage{amsmath,amsfonts,graphicx,framed}
\usepackage{amssymb}
\usepackage{amsthm}
\usepackage{fancyhdr}
\usepackage{float}
\usepackage{caption}
\usepackage{comment}
\usepackage[utf8]{inputenc}
\usepackage{authblk}
\usepackage[mathscr]{euscript}

\theoremstyle{plain}
\newtheorem{thm}{Theorem}

\newtheorem{prop}{Proposition}

\theoremstyle{definition}

\newtheorem{rem}{Remark}
\newtheorem{ex}{Example}

\numberwithin{equation}{section}

\newcommand{\pd}{\,\partial}

\newcommand{\mbb}{\mathbb}
\newcommand{\pdf}[2][]{\,\frac{\partial #1}{\partial #2}}
\newcommand{\ep}{\varepsilon}
\newcommand{\re}{\mbb R}
\newcommand{\al}{\alpha}

\newcommand{\ms}{\mathscr}

\newcommand{\eqal}[1]{\begin{equation}\begin{aligned}#1\end{aligned}\end{equation}}

\newcommand{\deq}{\,\dot = \,}

\begin{document}

\title{Quasi-Noether systems and quasi-Lagrangians}
\author{V. Rosenhaus*}                    
\author{Ravi Shankar}

\affil{\footnotesize{* Department of Mathematics and Statistics, California State University, Chico, CA 95929, USA, vrosenhaus@csuchico.edu}}
\affil{\footnotesize{Department of Mathematics, University of Washington, Seattle, WA 98195, USA, 
shankarr@uw.edu}}  
\date{}

\maketitle

\begin{abstract}
We study differential systems for which it is possible to establish a correspondence between symmetries and conservation laws based on Noether identity: quasi-Noether systems. We analyze Noether identity and show that it leads to the same conservation laws as Lagrange 
(Green-Lagrange) identity.
We discuss quasi-Noether systems, and 
some of their properties, and generate classes of quasi-Noether differential equations of the second order.  
We next introduce a more general version of quasi-Lagrangians which allows us to extend Noether theorem. Here, variational symmetries are only sub-symmetries, not true symmetries. 
We finally introduce the critical point condition for evolution equations with a conserved integral, demonstrate examples of its compatibility, and compare the invariant submanifolds of quasi-Lagrangian systems with those of Hamiltonian systems.
\end{abstract}

\section{Introduction}

\smallskip

For variational systems the relation between symmetries of the Lagrangian function and conservation laws was known from the classical Noether result 
\cite{Noether}. It was shown that there is one-to-one correspondence between variational symmetries (symmetries of variational functional) and local conservation laws 
of a differential system, \cite{Olver}.
\smallskip

In this paper, we study differential systems that allow a Noether-type association between its conservation laws and symmetries (quasi-Noether systems). Our approach is based on the Noether operator identity \cite{Rosen} that relates 
the infinitesimal transformation operator to the Euler and divergence operators.  The Noether operator identity has been shown to provide a Noether-type relation between symmetries and conservation laws not only for Lagrangian systems, but also for a large class of differential systems that may not have a well-defined variational functional, see \cite{Rosenhaus94}, \cite{Rosenhaus96}.
\smallskip
Noether operator identity was also demonstrated to allow derivation of extension of Second Noether Theorem  for non-Largangian systems possessing infinite symmetry algebras parametrized by arbitrary functions of all independent variables, \cite{Rosenhaus15}. These infinite symmetry algebras were shown to lead to differetial identities between the equations of the original differential system and their derivatives.
\smallskip
Recently Noether identity was used to generate relations between sub-symmetries \cite{RS2018} and corresponding local conservation laws \cite{RS2019} for quasi-Noether systems.

\medskip

In this paper, we analyze quasi-Noether systems and some of their properties. 
In Section \ref{sec:var}, we review known correspondence between symmetries and conservation laws for variational systems from the standpoint of the Noether identity.   
In Section \ref{sec:quasi}, we discuss this correspondence for a more general class of differential systems (quasi-Noether systems) that includes non-variational problems. We review conservation laws obtained with the use of approach based on the Noether identity and compare them with the results obtained from the Lagrange (Green-Lagrange) identity. 
We also find the class of quasi-Noether evolution equations of the second order, and quasi-linear equations of the second order.  In Section \ref{sec:qLagr}, we 
generalize the concept of quasi-Lagrangian, and use it to prove a new extension of the Noether theorem to non-Lagrangian systems.  
In this approach, variational symmetries of quasi-Lagrangians are only sub-symmetries, and need not be symmetries.  We give 
an example where all 
lower 
conservation laws are generated in this way, but where the previous correspondence fails.  As a geometric application, we compare the invariant submanifolds of quasi-Noether systems to those of Hamiltonian systems and show 
that 
they behave in an ``opposite" way. 
To address this fact we introduce the notion of a critical point of a conserved quantity, and 
demonstrate examples of the compatibility of the critical point condition with the time evolution of the PDE system.

\section{Symmetries and conservation laws of variational systems}
\label{sec:var}

Let us briefly outline the approach we follow.

\[
\Delta^a(x,u,u_{(1)},u_{(2)},\dots ) = 0, \qquad a = 1,\dots,q,
\]
we mean a divergence expression
\eqal{\label{CL}
D_i K^i(x,u,u_{(1)},u_{(2)},\dots ) \doteq 0,
}
that vanishes on all solutions of the original system; we denote this type of equality by ($\doteq$).  Here, $x=(x^1,x^2,\dots,x^p)$ and $u=(u^1,u^2,\dots,u^q)$ are the tuples of 
independent and dependent variables, respectively; 
$u_{(r)}$ is the tuple of $r$th-order derivatives of~$u$, $r=1,2,\dots$;
$\Delta^a$ and $K^i$ are differential functions, i.e. smooth functions of~$x$, $u$ and a finite number of derivatives of~$u$ (see \cite{Olver});
$i,j = 1,\dots, p$, $a=1,\dots,q$.  We assume summation over repeated indices.

We let
\begin{align*}
D_i=\partial_i+u^a_i\partial_{u^a}+u^a_{ij}\partial_{u^a_{ij}}+\dots=\partial_i+u^a_{iJ}\partial_{u^a_J}
\end{align*}
be the $i$--th total derivative, $1\le i\le p$, the sum extending over all (unordered) multi-indices $J=(j_1,j_2,...,j_k)$ for $k\ge 0$ and $1\le j_k\le p$.

Two conservation laws $K$ and $\tilde K$ are equivalent if they differ by a trivial conservation law \cite{Olver}.  A conservation law $D_i P^i \doteq 0 $ is trivial if a linear combination of 
two kinds of triviality is taking place: 1. The $p$-tuple $P $ vanishes on the solutions of the original system: $P^i \doteq 0$. 2. The divergence identity is satisfied  for any point $[u]=(x,u_{(n)})$  in the jet space (e.g. div rot $u=0$).
									
We consider smooth functions $u^a=u^a(x)$ defined on an 
open subset $D\subset\re^p$.  
Let
\[
S ={\int_D} {L(x,u, u_{(1)}, \dots)\: d^{p}x} 
\]											
be the action functional, where $L$ is the Lagrangian density. The equations of motion are
\begin{equation}\label{Rosenhaus:equation2}	
E_a(L)
 \equiv \Delta^a(x,u, u_{(1)}, \dots) = 0,\quad 1\le a\le q,
\end{equation}										
where 
\begin{equation}\label{Rosenhaus:equation3}	
E_a 
= \frac{\partial }{\partial u^a} - \sum\limits_i {D_i \frac{\partial }{\partial u_i^a } }
+ \sum\limits_{i \leqslant j} {D_i } D_j \frac{\partial }{\partial u_{ij}^a } + \cdots 
\end{equation}
is the $a$-th Euler (Euler--Lagrange) operator (variational derivative).  
We call the tuple $E=(E_1,\dots,E_n)$ the Euler operator.  In the notation of \cite{Olver}, we could 
give 
it the following form:
\begin{align}\label{euler_op}
E_a=(-D)_J\frac{\partial}{\partial_{u^a_J}},\hspace{4 mm}1\le a\le q,
\end{align}
The operator $(-\mathrm D)_J$ is defined here as $(-D)_J=(-1)^kD_J=(-D_{j_1})(-D_{j_2})\cdot\cdot\cdot(-D_{j_k})$.
The operator $E_a$ annihilates total divergences.  

Consider an infinitesimal (one-parameter)
transformation with the canonical infinitesimal operator
\begin{align}\label{Rosenhaus:equation4}	
X_\alpha = \alpha ^a &\frac{\partial }{\partial u^a} + \sum\limits_i {(D_i 
\alpha ^a)\frac{\partial }{\partial u_i^a }} + \sum\limits_{i \leqslant j} {(D_i } 
D_j \alpha ^a)\frac{\partial }{\partial u_{ij}^a } + \cdots\, = \, \mathrm (D_{J}{\alpha}^a) \partial_{u^a_J}.
\end{align}
where 
$\alpha ^a = \alpha ^a (x,u, u_{(1)}, \dots$), and the sum is taken over all (unordered) multi-indices $J$.										
The variation of the functional $S$ under the transformation with operator 
$X_\alpha $ is
\begin{equation}								\label{Rosenhaus:equation5}	
\delta S = \int_D{X_\alpha L  \,d^{p}x}\,.
\end{equation}										
$X_\alpha $ is a variational (Noether) symmetry if 
\begin{equation}								\label{Rosenhaus:equation6}	
X_\alpha L =  D_i M^i,
\end{equation}										
where $M^i=M^i (x,u,u_{(1)},\dots)$ are smooth functions of their arguments. 
The Noether (operator) identity \cite{Rosen} (see also, e.g. \cite{Ibragimov1985} or \cite{Rosenhaus94}) relates the operator $X_\alpha$ to $E_a$,
\begin{equation} 								\label{Rosenhaus:equation7}	
X_\alpha = \alpha^aE_a +{D_i R^{i} }, 		
\end{equation}										
\begin{equation}	\label{Rosenhaus:equation8}	
R^{i} = \alpha ^a\frac{\partial }{\partial u_i^a } + \left\{ 
{\sum\limits_{k \geqslant i} {\left( {D_k \alpha ^a} \right) - \alpha 
^a\sum\limits_{k \leqslant i} {D_k } } } \right\}\frac{\partial }{\partial 
u_{ik}^a }+ \cdots.									
\end{equation}
The expression for $R^{i}$ can be presented in a more general form \cite{Rosen,Lunev}:
\begin{align}\label{Ri}
R^i=(D_K{\alpha}^a)\,(-D)_{J}{\partial}_{u^a_{iJK}},
\end{align}
where $J$ and $K$ sum over multi-indices.

Applying the identity \eqref{Rosenhaus:equation7} with 
\eqref{Rosenhaus:equation8} to $L$ and using \eqref{Rosenhaus:equation6},	  
we obtain
\begin{equation}								\label{Rosenhaus:equation9}	
D_i (M^i - R^i L) = \alpha ^a\Delta ^a, 
\end{equation}										
which on the solution manifold ($\Delta = 0$, $D_i \Delta = 0$, \dots ) 
\begin{equation}								\label{Rosenhaus:equation10}	
D_i (M^i - R^i L)\doteq 0,   
\end{equation}										
leads to the statement of the First Noether Theorem: any 
one-parameter variational symmetry transformation with infinitesimal 
operator $X_\alpha$ \eqref{Rosenhaus:equation4} gives 
rise to 
the 
conservation law \eqref{Rosenhaus:equation10}.
\medskip

Note that Noether \cite{Noether} used the identity \eqref{Rosenhaus:equation9} and not the operator identity \eqref{Rosenhaus:equation7}.  The first mention of the Noether operator identity  \eqref{Rosenhaus:equation7}, to our knowledge, was made in \cite{Rosen}.
 
\smallskip

We next consider differential systems that may not have 
well-defined Lagrangian  functions.

\section{Symmetries and conservation laws of quasi-Noether systems}
\label{sec:quasi}

For a general differential system, a relationship between symmetries and 
conservation laws is unknown. In \cite{Rosenhaus94} and \cite{Rosenhaus96}, 
an approach based on the Noether operator identity \eqref{Rosenhaus:equation7} was suggested to relate symmetries to
conservation laws for a large class of differential systems that may not 
have well-defined Lagrangian functions.
In the current paper, we will follow this approach.

\subsection{Approach using the Noether operator identity}

We consider $q$ smooth functions 
$u=(u^1,u^2,...,u^q)$ of $p$ independent variables $x=(x^1,x^2,...,x^p)$ 
defined on some nonempty open subset of $\mathbb{R}^p$.
Consider a system of $n$ $\ell$th order differential equations 
$\Delta=(\Delta^1,\Delta^2,...,\Delta^n)$ for functions $u$ :
\begin{align}\label{sys}								
\Delta^a(x,u,u_{(1)},u_{(2)},\dots, u_{(l)})=0, \quad a=1,2,...,n.
\end{align}
Here, each $\Delta^a(x,u,u_{(1)},u_{(2)},\dots, u_{(l)} )$,  $a=1,2,...,n$ 
is a smooth function of $x$, $u$, and all partial derivatives of each 
$u^v$, ($v=1,\dots,q$) with respect to the $x^i$ ($i=1,\ldots,p$)
up to the $\ell$th order (differential function \cite{Olver}). 
We assume that $\Delta$ \eqref{sys} is normal, totally nondegenerate system
(locally solvable at every point, and of maximal rank) \cite{Olver}. 

Let 
$\Delta(x,u,u_{(1)},u_{(2)},\dots, u_{(l)} )\equiv \Delta[u]\equiv \Delta$, 
$x_J=(x_{j_1},x_{j_2},...,x_{j_k})$ and $u^v_J$ be partial derivatives, where 
$J=(j_1,j_2,...,j_k)$. Applying the Noether operator identity \eqref{Rosenhaus:equation7} 
to a combination of original equations with some coefficients (differential operators) $\beta^a$ $\Delta^a$, we obtain
\begin{equation}								\label{Rosenhaus:equation12}	
X_\alpha( {\beta^a\Delta^a})= \alpha^v E_v (\beta^a\Delta^a) +{D_{i} R^i(\beta^a\Delta^a)},
\qquad a =1, \dots n, \quad v=1, \dots, q, \quad i=1, \dots, p.
\end{equation}
If our system \eqref{sys} allows the existence of coefficients $\beta^a$
(\textit{cosymmetries}) such that 
\begin{equation}  \label{quasi-Noether}										
E_v (\beta^a\Delta^a)\doteq 0, \qquad a= \dots, n, \quad v=1, \dots, q
\end{equation}										
on the solution manifold 
($\Delta = 0,   D_i \Delta = 0, \dots $), then according to 
\eqref{Rosenhaus:equation12}, each symmetry of the system $X_\alpha$ will 
lead to a local conservation law (see \cite{Rosenhaus96}, see also \cite{Rosenhaus94})
\begin{equation}								\label{Rosenhaus:equation14}	
D_i R^i (\beta^a\Delta^a)\doteq 0.   
\end{equation}	
for any differential systems of class \eqref{quasi-Noether}.  In \cite{Rosenhaus94}, 
the quantity $\beta^a\Delta^a$ was referred to as an alternative Lagrangian. 

Let us note that the correspondence between symmetries and local conservation laws defined 
above for differential systems without well-defined Lagrangian functions may not be one-to-one
or onto, 
as in the case of variational symmetries and local conservation laws \cite{Olver}, and non-trivial symmetries may lead to trivial conservation laws.  
If $\beta$ generates a conservation law, i.e. $E(\beta\cdot\Delta)\equiv 0$, then it was shown in \cite{Anco2017} that translation symmetries $\al=a^iu_i, a=const.$ lead to trivial conservation laws if $\beta_x=0$.  

\medskip
In general, the nontriviality 
of a conservation law 
is determined by the characteristic.  To compute it explicitly, let $X_\al(\beta\cdot\Delta)=A\cdot \Delta$ and $E(\beta\cdot\Delta)=B\cdot \Delta$ for $A,B$ differential operators.  Then the Noether identity and integration by parts yield
\eqal{
A\cdot \Delta=X_\al(\beta\cdot\Delta)&=\al\cdot E(\beta\cdot\Delta)+div\\
&=\al\cdot (B\cdot\Delta)+div\\
&=\Delta\cdot(B^*\cdot \al)+div,
}
where $(B^*)_{va}=(-D)_I\circ B_{avI}$ is the adjoint operator.  Integrating by parts the LHS and rearranging yields overall
\eqal{
\,[B^*\cdot\al-A^*(1)]\cdot\Delta=div,
}
where the divergence is equivalent to \eqref{Rosenhaus:equation14} on solutions, and $A^*(1)_v=(-D)_IA_{vI}$.  
\qed

Thus, the conservation law obtained from Noether identity and corresponding to symmetry $X_{\alpha} $ is nontrivial if

\begin{align}
B^*\cdot\al-A^*(1) \ne 0  
\end{align}
 on solutions of the original system \eqref{sys}.
\smallskip

The condition \eqref{quasi-Noether} can be written in a somewhat 
more general form \cite{Rosenhaus96}
\begin{equation}								\label{Rosenhaus:equation17}	
E_v (\beta^{ac}\Delta^a) \doteq 0, \qquad a, c = 1, \dots n, \quad v=1, \dots, q.
\end{equation}

In terms of the Fr\'echet derivative operator $ D_{\Delta} $ and its adjoint 
$D^{*} _{\Delta}$ \cite{Olver}
\begin{align}\label{Rosenhaus:equation18}
\left(\mathrm D_{\Delta}\right)_{av} & =   \sum\limits_{J} \frac {\partial {\Delta^a}}
{\partial {u^v_J}} D_J, \\
\left(\mathrm D^{*}_{\Delta}\right)_{av} & {\beta^a}  = \sum\limits_{J} 
(-1)^kD_{J}\left( \frac {\partial {\Delta^a}}
{\partial {u^v_J}} \beta^a\right),
\end{align}
the condition \eqref{quasi-Noether} was shown in \cite{Rosenhaus97} to 
be related to the condition of self-adjointness of the operator $ \mathrm D_{\Delta} 
$ (generalized Helmholtz condition)
\begin{align}\label{Rosenhaus:equation19}
\mathrm D^{*} _{\Delta} - \mathrm D_{\Delta} = 0,
\end{align}
this relationship being:
\begin{equation}								\label{Rosenhaus:equation20}	
E_v (\beta^a\Delta^a) \doteq (\mathrm D^{*}_{\Delta}-\mathrm D_{\Delta})_{av} \beta^a, 
\qquad a = 1, \dots n, \quad v=1, \dots, n.							
\end{equation}

Expression \eqref{Rosenhaus:equation20} provides a relationship between the 
condition \eqref{quasi-Noether} for a system to be quasi-Noether and the existence 
of a variational functional for a transformed differential 
system.

The condition \eqref{quasi-Noether} can be considered as defining quasi-Noether systems, see also \cite{Rosenhaus15}. 
A system \eqref{sys} is \emph{quasi-Noether} if there exist functions (differential operators) $\beta^a$ such that the condition \eqref{quasi-Noether} is satisfied. In \cite{Rosenhaus94}, 
the quantity $\beta^v\Delta_v$ was referred to as an alternative Lagrangian. 
 If coefficients $\beta^a$ generate a conservation law, they are 
related to adjoint symmetries \cite{Sarlet1987}, and referred to as 
characteristics of a corresponding conservation law \cite{Olver}, generating functions \cite{Vinogradov} or multipliers \cite{AncoBluman1997}. 
 In general, we call them cosymmetries
 \cite{Vinogradov} 
if they satisfy \eqref{quasi-Noether}.

\medskip

It should be noted that the condition \eqref{quasi-Noether} (for a system to be quasi-Noether and possess a correspondence between symmetries and conservation laws) was obtained earlier within an alternative approach based on the Lagrange identity.
In \cite{Vladimirov1980}, the classical Lagrange identity (Green's formula)
was used for generating conservation laws for a linear differential system.  A condition for the existence of a certain conservation law written in terms of an adjoint differential operator was presented in \cite{Vinogradov} within a general framework of algebraic
geometry. In \cite{Zharinov1986}, a correspondence 
between symmetries and conservation laws  based on Green's formula and a  condition similar to 
\eqref{quasi-Noether} was obtained for evolution systems. For the case of point transformations, this correspondence was discussed in \cite{Caviglia1986}. In \cite{Sarlet1987}, the geometric meaning 
of this condition was discussed for mechanical systems. In \cite{Lunev}, a condition that can be reduced to \eqref{quasi-Noether} was presented for a general differential system.  

\medskip

Note that the condition  \eqref{quasi-Noether} played a key role in later developed direct 
method \cite{AncoBluman1997} and the nonlinear self-adjointness approach \cite{Ibragimov2011}.

\medskip
\subsection{Approach using Lagrange identity}
Let us briefly describe the alternative approach based on the Lagrange identity and compare the results with those of the Noether identity approach.

The well-known Lagrange identity is as follows:
\begin{align}\label{lag}
\beta^a(\mathrm D_\Delta)_{av}\alpha^v-\alpha^v(\mathrm D^*_\Delta)_{av}\beta^a=D_iQ^i[\alpha,\beta,\Delta],
\end{align}
where $\mathrm D_\Delta$ is the Fr\'echet derivative of $\Delta$, and $\mathrm D^*_\Delta$ is its adjoint. 
An explicit expression for the trilinear fluxes $Q[\alpha,\beta,\Delta]$ was given in \cite{Lunev}:
\begin{align*}
Q^i[\alpha,\beta,\Delta]=(-1)^{|J|}D_K\alpha^vD_J\left(\beta^a\partial_{u^v_{iJK}}\Delta^a\right),\quad 1\le i\le p.
\end{align*}
Using the fact that 
\begin{align} 
(\mathrm D_\Delta)_{av}\alpha^v = X_{\alpha}\Delta^a,
\end{align}
(see \eqref{Rosenhaus:equation4}, and \eqref{Rosenhaus:equation18}) we can express \eqref{lag} as follows:
\begin{align}\label{lagS}
\beta^aX_{\alpha}\Delta^a-\alpha^v(\mathrm D^*_\Delta)_{av}\beta^a=D_iQ^i[\alpha,\beta,\Delta].
\end{align}

If, for a given system $\Delta$, there exist functions $\beta^a$ and operators $\Upsilon^{va}=\upsilon^{vaJ}D_J$ such that the following relationships hold:
\begin{align}\label{dstar}
(\mathrm D^*_\Delta)_{av}\beta^a=\Upsilon^{va}\Delta^a,\quad 1\le v\le q,
\end{align}
then $(\mathrm D^*_\Delta)_{av}\beta^a\,\dot{=}\,0$.  Thus, for each symmetry $X_\alpha$ of the system $\Delta=0$, ($X_\alpha\Delta^a=\Lambda^{ab}\Delta^b$ for some operators $\Lambda^{ab}=\lambda^{abJ}D_J$), equation \eqref{lagS} provides a corresponding conservation law:
\begin{align}
\begin{split}
D_iQ^i[\alpha,\beta,\Delta]&=(\beta^a\Lambda^{ab}-\alpha^v\Upsilon^{va})\Delta^a\,\dot{=}\,0.
\end{split}
\end{align}

We now show that 
the correspondence between symmetries and conservation  laws in terms of Lagrange identity is equivalent to the one using the Noether identity and leads to the same conservation laws. Indeed, using the product rule \cite{Olver}:
\begin{align}
E_v(\beta^a\Delta^a)=(\mathrm D_{\Delta}^*)_{av}\beta^a+(\mathrm D_\beta^*)_{av}\Delta^a,\quad 1\le v\le q,
\end{align}
in
\eqref{lagS} gives:
\begin{align*}
\beta^aX_\alpha\Delta^a-\alpha^vE_v(\beta^a\Delta^a)+\alpha^v(\mathrm D_\beta^*)_{av}\Delta^a=D_iQ^i.
\end{align*}
Using
$\beta^aX_\alpha\Delta^a=X_\alpha(\beta^a\Delta^a)-\Delta^aX_\alpha\beta^a$ 
we obtain 
\begin{align}\label{lagN}
X_\alpha(\beta^a\Delta^a)=\alpha^vE_v(\beta^a\Delta^a)+D_iQ^i[\alpha,\beta,\Delta]+(X_\alpha\beta^a-\alpha^v(\mathrm D^*_\beta)_{av})\Delta^a.
\end{align}
The last term in \eqref{lagN} is a total divergence by \eqref{lag}: 
\begin{align*}
\Delta^aX_{\alpha}\beta^a-\alpha^v(\mathrm D^*_\beta)_{av}\Delta^a
&
=\Delta^a(\mathrm D_\beta)_{av}\alpha^v-\alpha^v(\mathrm D^*_\beta)_{av}\Delta^a 
=D_iQ^i[\alpha,\Delta,\beta].
\end{align*} 
Thus, \eqref{lagN} provides the same result as the Noether identity:
\begin{align}
X_\alpha(\beta^a\Delta^a)=\alpha^vE_v(\beta^a\Delta^a)+D_i(Q^i[\alpha,\Delta,\beta]+Q^i[\alpha,\beta,\Delta]).
\end{align}

\bigskip

Let us show that the condition used in direct 
method \cite{AncoBluman1997}
\begin{equation} 								\label{Rosenhaus:equation21}	
(\mathrm D^{*}_{\Delta})_{av} \beta^a\,\dot{=}\,0						
\end{equation}
is equivalent to the quasi-Noether condition \eqref{quasi-Noether}.  Indeed,
using the identity \cite{Olver}
\begin{equation} 								\label{Rosenhaus:equation22}	
E_v (\beta^{a}\Delta^a)=(\mathrm D^{*}_{\Delta})_{av} \beta^a + (\mathrm D^{*}_{\beta})_{av} \Delta^a,
\end{equation}
and the fact that 
\begin{equation}								\label{Rosenhaus:equation23}	
(\mathrm D^{*}_{\beta})_{av} \Delta^a \doteq 0,
\end{equation}
we obtain condition \eqref{Rosenhaus:equation21}
\begin{equation}								\label{Rosenhaus:equation24}	
E_v (\beta^{a}\Delta^a) \doteq (\mathrm D^{*}_{\Delta})_{av} \beta^a \doteq 0.
\end{equation}
An alternative key expression of the direct method
\begin{equation}								\label{Rosenhaus:equation25}	
E_v (\beta^{a}\Delta^a)= 0
\end{equation}
is a special case of \eqref{quasi-Noether}. Note also that the direct method aims at the generation of conservation laws for a differential system without regard to its symmetries while the goal of both approaches above taken earlier was to establish a correspondence between symmetries and conservation laws.


\subsection{Quasi-Noether systems}
As noted above the condition \eqref{quasi-Noether} 
(or \eqref{Rosenhaus:equation17}) determines \textit{quasi-Noether systems}. 
 
It can be shown that the general case of differential systems \eqref{sys} satisfying the condition \eqref{quasi-Noether}, with $\beta^a$ being differential operators
\begin{equation}
\beta^{a}=\beta^{aJ}[u]D_J, \qquad D_J=D_{j_1 }D_{j_2}\cdot\cdot\cdot D_{j_k}, \quad
D_{j_r}= \frac {d}{dx_{j_r}},\quad r=1,2,. . ., k,				\label{Rosenhaus:equation26}
\end{equation} 
can be reduced to the case of the condition \eqref{quasi-Noether} with the 
$\beta^{aJ}$ being differential functions. Indeed, (see \cite{Olver}) 
\begin{align}									\label{Rosenhaus:equation27}				
\beta^a \Delta^a = \beta^{aJ}D_{J}\Delta^a = 
D_{J}\left(\beta^{aJ} \Delta^a \right) + \left((-\mathrm D)_{J}\beta^{aJ}\right) \Delta^a.	
\end{align}
Since the contribution of the first term in the RHS, being a total divergence, is zero after applying the Euler operator, we find that
\begin{equation}								\label{Rosenhaus:equation28}	
E_v (\beta^a \Delta^a) \doteq E_v(\bar{\beta}^a \Delta^a)\doteq 0, 
\qquad a = 1, \dots n, \quad v=1, \dots, q,
\end{equation}
where each $\bar{\beta}^a = \left((-\mathrm D)_{J}\beta^{aJ}\right)$ is a differential function. 

\medskip

In the following theorem, we prove that local conservation laws exist only 
for quasi-Noether systems.

\begin{thm}[Quasi-Noether systems and conservation laws]
\label{thm:quasi-Noether}
Any differential system \eqref{sys} that possesses local conservation laws is quasi-Noether.
A quasi-Noether system that admits continuous symmetries 
in general, possesses local conservation laws.  
\end{thm}

\begin{proof}
Assume that the system \eqref{sys} has some local conservation law
\begin{align}
D_{i}P^{i}= \gamma^a \Delta^a,
\end{align}
where $P^i$ and $\gamma^a$ are differential functions.
Then
\begin{equation}								\label{Rosenhaus:equation29}	
E_v (\gamma^a \Delta^a)= 0, \qquad a =1, \dots n, \quad v=1, \dots, q, \quad i=1, \dots, p,
\end{equation} 
which shows that the system is quasi-Noether.

\medskip

\noindent Suppose now that the quasi-Noether system \eqref{quasi-Noether} possesses a continuous symmetry with the infinitesimal operator 
$ X_\alpha $ \eqref{Rosenhaus:equation4}. Rewriting 
condition \eqref{quasi-Noether} in the form
\begin{align}\label{quasiN}
E_v (\beta^a\Delta^a) = \Gamma^{va}\Delta^a \doteq 0, \qquad v = 1, \dots q,
\end{align}
where the sum is taken over $1\le a \le n$, $\beta^a$ are differential functions, and $\Gamma^{va}=\Gamma^{vaJ}D_J$ are differential operators. 
Multiplying the equation \eqref{quasiN} by $\alpha^v$ and taking a sum over $v$, we obtain
\begin{align}\label{quasiNN}
\alpha^v E_v (\beta^a\Delta^a) = \alpha^v \Gamma^{va}\Delta^a.
\end{align}
Applying the Noether identity \eqref{Rosenhaus:equation7} in the LHS, we obtain  
\begin{equation}								\label{Rosenhaus:equation30}	
\left(X_\alpha -{D_{i} R^i} \right) (\beta^a\Delta^a)= \alpha^v \Gamma^{va}\Delta^a, 
\qquad a =1, \dots n, \quad v=1, \dots, q, \quad i=1, \dots, p.
\end{equation}
Since $ X_\alpha $ is a symmetry operator for the system $\Delta^a$, we obtain
\begin{equation}								\label{Rosenhaus:equation31}	
{D_{i} R^i(\beta^a\Delta^a)} \doteq 0,
\qquad a =1, \dots n, \quad v=1, \dots, q, \quad i=1, \dots, p,
\end{equation}
and therefore, a quasi-Noether system with continuous symmetries 
in general, possesses local conservation law \eqref{Rosenhaus:equation31}.
\end{proof}

\medskip
Note that in some cases the conservation law \eqref{Rosenhaus:equation31} may be trivial.
\smallskip

Note 
also that condition \eqref{quasi-Noether} allows one to find and classify quasi-Noether equations of a certain type.  For example, it can be shown that equations of the following class:
\begin{align}
u_t=u_{xx}+u_x^n,\quad n=0,1,2,3,...
\end{align}
possess conservation law(s) only for $n=0,1,2$.  Correspondingly, for $n\ge 3$, such equations are not quasi-Noether and, hence, do not admit any conservation laws.

\subsection{Classes of Quasi-Noether equations}
The quasi-Noether class 
of equations for which it is possible to establish a Noether type 
correspondence between symmetries and conservation laws is quite large, and covers practically all interesting differential systems of mathematical physics, and all 
systems possessing conservation laws. Many examples of equations of the class were given in \cite{Rosenhaus94}, \cite{Rosenhaus96}, 
and \cite{Rosenhaus97}.
Quasi-Noether systems include all differential systems in the form of conservation laws, e.g. KdV, mKdV, Boussinesq, Kadomtsev-Petviashvili equations, nonlinear wave and heat equations, Euler equations, and Navier-Stokes equations; as well as the homogeneous Monge-Ampere equation, and its multi-dimensional analogue, see \cite{Rosenhaus94}. In \cite{Rosenhaus15}, the approach based on the Noether identity was applied to quasi-Noether systems possessing infinite symmetries involving arbitrary functions of all independent variables, in order to generate an extension of the Second Noether theorem for systems that may not have well-defined Lagrangian functions.

\medskip
\subsubsection{Evolution equations}
Consider evolution equations of the form 
\begin{align}
u_t=A(u,u_x)u_{xx}+B(u,u_x).
\end{align}
Assuming $\beta=\beta(t,x,u,u_x,u_{xx},\dots,u_{nx})$, and requiring our equation to be quasi-Noether \eqref{quasi-Noether} we obtain the following condition 
\begin{align}
E_u[\beta u_t-\beta(Au_{xx}+B)]\,\dot{=}\,0.
\end{align}
It can be shown that 
$\beta=\beta(t,x,u)$.  Moreover, it can be shown that for this system, the above equality holds for all $u$ $(=)$ rather than strictly solutions $(\,\dot{=}\,)$.

\smallskip
We obtain the following classes of equations:
\medskip

1.
\begin{align}
\begin{split}
A(u,u_x)&=\alpha H_{u_x}/S'(u),\\
B(u,u_x)&=\alpha (u_x H_u-S(u))/S'(u),\\
\beta(t,u)&=e^{\alpha t}S'(u),
\end{split}
\end{align}
where $H(u,u_x)$ and $S(u)$ are arbitrary functions, and $\alpha$ is a constant.  

In this case the equation $\Delta=u_t-Au_{xx}-B$ is obviously, quasi-Noether since the left hand side turns into total divergence upon multiplication by $\beta$:
\begin{align*}
\beta\Delta=D_t[e^{\alpha t}S(u)]-\alpha D_x[e^{\alpha t}H(u,u_x)].
\end{align*}

\smallskip

2.
\begin{align}
\begin{split}
A(u,u_x)&=(u_xG_{u_x}-G)/u_x^2,\\
B(u,u_x)&=G_u+(b/u_x+c'(u))G+h+u_x(a+h_u+hc_u)/b,\\
\beta(t,x,u)&=\exp(a t+bx+c(u)),
\end{split}
\end{align}
where $G(u,u_x)$, $c(u)$, and $h(u)$ are arbitrary functions, and $a$ and $b$ are constants.


a. The special case $G(u,u_x)=u_x^2$, $c(u)=u$, $h(u)=0$, and $a=-b^2$, leads to the following equation:
\begin{align} \label{A1}
u_t-u_{xx}-u_x^2=0.
\end{align}
Multiplication by $\beta=\exp(bx-b^2t+u)$ turns the LFS of equation \eqref{A1} into a total divergence:
\begin{align*}
D_t(\beta u)-D_x[\beta(u_x-b)]=0.
\end{align*}
\smallskip
b. Choosing $a=-b^2, \: c'(u)=0, \: G(u,u_x)= {u_x}^2+pu_x, \: p=const, \:h(u)=1-bp$ we obtain heat equation, $A=1, \: B=1$ 
\begin{align}
u_t=u_{xx}
\end{align}
\smallskip
c. Choosing $a=-b^2, \: c'(u)=0, \: G(u,u_x)= {u_x}^2+u^2u_x/2, \:h(u)=-bu^2/2 $ we obtain Burgers equation, $A=1, \: B=-uu_x$ 
\begin{align}
u_t=u_{xx}-uu_x.
\end{align}
\smallskip
d. Choosing $a=-b^2=-1, \: c'(u)=0, \: G(u,u_x)= {u_x}^2+(u-u^2)u_x, \:h(u)=0 $ we obtain Fisher equation, $A=1, \: B=u-u^2$ 
\begin{align}
u_t=u_{xx}+u(1-u).
\end{align}

3. 
\begin{align}
\begin{split}
A(u,u_x)&=a,\\
B(u,u_x)&=au_x^2(\Phi_{uu}+b\Phi_u^2)/\Phi_u+1/\Phi_u+\epsilon u_x,\\
\beta(t,x,u)&=v(t,x)\exp(-bt+b\Phi),
\end{split}
\end{align}
where $\Phi(u)$ is an arbitrary function, $a$ and $b$ are constants, and $v(x,t)$ is a solution of the following equation:
\begin{align*} \label{A2}
v_t-bv_x+av_{xx}=0.
\end{align*}
\smallskip
4.
\begin{align}
\begin{split}
A(u,u_x)&=a,\\
B(u,u_x)&=au_x^2\ell_u+\epsilon u_x,\\
\beta(t,x,u)&=v(x,t)\,e^\ell,
\end{split}
\end{align}
where $\ell(u)$ is an arbitrary function, and $v$ satisfies the same linear equation \eqref{A2}.
\smallskip

5.
\begin{align}
\begin{split}
A(u,u_x)&=G_u/\Phi_u,\\
B(u,u_x)&=u_x^2(G_{uu}+\tau G_u^2+\epsilon G_u\Phi_u)/\Phi_u+u_x(\delta+2\sigma G_u/\Phi_u)+1/\Phi_u,\\
\beta(t,x,u)&=\exp[-\epsilon t+\sigma(x+\delta t)+\tau G+\epsilon \Phi]\cosh[(x+\delta t-\mu)\sqrt{\sigma^2-\tau}]\Phi_u,
\end{split}
\end{align}
where $\Phi(u)$ and $G(u)$ are arbitrary functions, and $\delta,\epsilon,\mu,\sigma,$ and $\tau$ are arbitrary constants.
\medskip

We do not pursue the remaining cases, since these necessarily involve dependence on $1/u_x$.  However, the above analysis gives insight into, for instance, equations of the following class:
\begin{align}
u_t=u_{xx}+u_x^n,\quad n=0,1,2,3,...
\end{align}
We have shown that such an equation possesses a conservation law only for $n=0,1,2$.  For $n\ge 3$, such equations are not quasi-Noether and, hence, do not admit conservation laws.

\medskip

\subsubsection{Quasi-linear equations}

Consider now quasi-linear equations of the form 
\begin{align} \label{QL}
\Delta=u_{tt}-A(u,u_x)u_{xx}-C(u,u_x)=0.
\end{align}
We require our equation to be quasi-Noether \eqref{quasi-Noether} 
\begin{align}\label{qNcond}
E_u[\beta (u_{tt}-Au_{xx}-C)]\,\dot{=}\,0.
\end{align}
Using the identity \eqref{Rosenhaus:equation22}
\begin{equation} 									
E_u (\beta\Delta)=(\mathrm D^{*}_{\Delta}) \beta + (\mathrm D^{*}_{\beta}) \Delta,
\end{equation}
we obtain 
\begin{equation}\label{cond}									
E_u (\beta\Delta) \doteq (\mathrm D^{*}_{\Delta}) \beta \doteq 0.
\end{equation}
For equations \eqref{QL} condition \eqref{cond} takes a form
\begin{align}\label{det_eq}
{D_{t}}^2\beta -{D_{x}}^2(\beta A) +D_x{(\beta u_{xx} A_{u_x})} &+D_t{(\beta u_{xx} A_{u_t})}+\\
D_x{(\beta C_{u_x})}+D_t{(\beta C_{u_t})} &-\beta C_u -\beta A_u u_{xx}  \nonumber =0.
\end{align}
Solving \eqref{det_eq} for $\beta=\beta(x,t,u)$, we obtain
\begin{align}
\begin{split}
A &=u_t L\;+\;M, \\
C &=u_t R\;+\;P\;+\;f(u,u_t),
\end{split}
\end{align}
where
\begin{align}
\begin{split}
R(u,u_x) =\int \left[\frac{\beta_x}{\beta}L+u_x L_u\right]\: du_x\\ 
P(u,u_x)\; =\; \int K(u,u_x)\: du_x,
\end{split}
\end{align}
and $ L=L(u,\; u_x),$  $K=K(u,\,u_x), \: f=f(u,u_t),\: \beta=\beta(x,t,u)$ are functions to be determined.

\smallskip
We get the following solutions:

\smallskip

1. 
\begin{align} \label{Class1}
\begin{split}
L &=\frac{1}{m}[2M_u+ lM_{u_x}+ u_x M_{uu_x}-K_{u_x}] \\
R &=\int \left[lL+u_x L_u\right]\: du_x=l^2 M+2 l u_x M_u+{u_x}^2 M_{uu}-l K - u_x K_u \: + \int {K_u du_x}, 
\end{split}
\end{align}
where  $l,m$ are abrbitrary constants,  $M(u,u_x), \: K(u,u_x)$ are arbitrary functions, $\beta = e^{lx+mt}$, and  $f(u,u_t)$ satisfies 
\begin{align}
-f_u+mf_{u_t}+u_t f_{uu_t}=-m^2.
\end{align}

\medskip

2. 
\begin{align}
\begin{split}
R &= L = 0 \\
K &=(l+qu_x) M+u_xM_u +\int [qM+M_{u}] du_x+r(u), 
\end{split}
\end{align}
where  $l,m,q$ are abrbitrary constants, function  $M(u,u_x)$ is arbitrary, $\beta = e^{lx+mt+qu}$, and  
\begin{align} \label{add}
f(u,u_t)=-qu_t^2+s(u)u_t+v(u), \:\:\:\: s(u)=   -m+\frac{qv(u)+v'(u)}{m}.
\end{align}

\smallskip

Examples:
\smallskip

a. If $ M(u,u_x) = c,$ we obtain
\begin{align}
u_{tt}= M u_{xx}+(lM+r(u))u_{x} +qM{u_x}^2 +R(u)+(-q{u_t}^2+su_t+v),
\end{align}
where $s(u)$ satisfies \eqref{add} and functions $r(u)$ and $R(u)$ are arbitrary.
\smallskip
Choosing $q=s=v=0, \: r=-l M$ we obtain
\begin{align} \label{main}
u_{tt}= M u_{xx}+R(u).
\end{align}
The class \eqref{main} includes Liouville equation ($R(u) =ke^{\lambda u}$)
\begin{align}
u_{tt} &= M u_{xx}+ke^{\lambda u},
\end{align}
and Sine-Gordon equation (($R(u) =ksin{\lambda u}$)
\begin{align}
u_{tt} &= M u_{xx}+ksin{\lambda u}.
\end{align}
\smallskip

b. Choosing $ M(u,u_x) =g(u), q=s=v=0, r(u)=-lg(u),$ we obtain
\begin{align} \label{importn}
u_{tt} &= g(u) u_{xx}+g'(u){u_x}^2.
\end{align}
The equation \eqref{importn} is a nonlinear wave equation.

\section{Quasi-Lagrangians}
\label{sec:qLagr}

For $\beta$ a cosymmetry, the quantity $L=\beta\cdot \Delta$ serves an analogous role to a Lagrangian.  In this section, we examine the applications of such a \textit{quasi-Lagrangian} $L$ in more detail.  We show 
that equations with a quasi-Lagrangian have a Noether correspondence on a subspace.  We then demonstrate the incompleteness of the symmetry/cosymmetry correspondence with an example.  We conclude with a comparison of the invariant submanifolds of quasi-Noether systems with those of Hamiltonian systems.

\subsection{A Noether correspondence}
Let us introduce the operator $T$ depending on $\beta$ such that $E(\beta\cdot\Delta)=T\cdot\Delta$.  If $\mbb D_\Delta$ is an endomorphism (same number of equations as dependent variables), then $T=\mbb D_\beta^*-\mbb D_\beta$ by \eqref{Rosenhaus:equation20}.

\medskip
We say an operator $T$ is nondegenerate if its restriction to $\Delta=0$ is not the zero operator.  If $\beta$ is a characteristic for a conservation law (generator), then $E(\beta\cdot\Delta)=0$, so $T=0$ (degenerate).  Thus $\beta$, in this section, will be a cosymmetry which is \textit{not} a characteristic.

\medskip
Recall that a functional $\ms L[u]=\int L[u]dx$ is an equivalence class modulo total divergences.  We say a smooth functional $\ms L[u]=\int L[u]dx$ is nondegenerate if it is nonzero and some representative $L[u]$ (hence all of them) does not vanish quadratically with $\Delta$.  We identify nondegenerate functionals with their affine terms: $L\sim L|_{\Delta=0}+\beta^I\cdot D_I\Delta$.  Mod a divergence, we then have $L\sim L_0+\beta\cdot\Delta$, where $\beta=(-D)_I\beta^I$ comes from integration by parts.  We will abuse notation in the below and identify $\ms L$ with $L$.  As a non-example, consider, for example, a second order scalar nonlinear PDE $\Delta=\Delta^1(u_{ij})$.  Then for $L=\Delta^2$, we have $E(L)=D_{ij}(A^{ij}\Delta)$, where $A_{ij}=2\pd \Delta/\pd u_{ij}$.

\medskip
The concept of \textit{sub-symmetry} was introduced in 
\cite{RS2018}. We say $X_\al$ generates a \textit{sub-symmetry} of $\Delta$ if $X_\al(T\cdot\Delta)|_{\Delta=0}=0$ for some nondegenerate $T$.  In the special case that $X_\al(T\cdot\Delta)|_{T\cdot\Delta=0}=0$, clearly $X_\al$ is an ``ordinary" symmetry of a sub-system $T\cdot\Delta$; the more general definition will not be needed here.  We note that in \cite{RS2019}, sub-symmetries of sub-systems generated by co-symmetries yield conservation laws through the mechanism in the previous section.  In the present work, we find a new appearance of sub-symmetries.

\medskip
We now present an extension of Noether theorem.  Every variational symmetry $X_\al$ of a quasi-Lagrangian $L$ corresponds to a conservation law characteristic $T^*\al$ in the image of $T^*$.  The variational symmetry is only a sub-symmetry of the PDE $\Delta$.  If $L=L_0+\beta\cdot\Delta$ and cosymmetry $\beta$ is also in the image of $T^*$, then a family of equations $\Delta+ker[T]$ shares variational symmetries, hence conservation laws generated by the quasi-Lagrangian.
\begin{thm}\label{thm:qL}
Let $\Delta$ be a normal, totally nondegenerate PDE system, and suppose there exist smooth and nondegenerate operator $T$ and functional $L$ such that $E(L)=T\cdot\Delta$.

\medskip
1. $X_\al L=div$, if and only if $(T^*\al)\cdot\Delta=div$.  

2. If $X_\al L=div$, then $\al$ is a sub-symmetry of $\Delta$.

\medskip
\noindent
Let $L=L_0+\beta\cdot\Delta$.  Define $\Delta_f=\Delta+f$ and $L_f=L_0+\beta\cdot \Delta_f$ for $f\in ker[T]$.

3. If $\beta\in im[T^*]$, $X_\al L_f=div$ if and only if $X_\al L=div$.

\end{thm}

\begin{proof}
1. The Noether theorem states $X_\al L=div$ if and only if $\al\cdot E(L)=div$, so $\al\cdot (T\cdot\Delta)=div$.  If we integrate by parts, we obtain $(T^*\al)\cdot\Delta=div$, as desired.

2. It is well known $X_\al L=div$ implies $X_\al$ is a symmetry of $E(L)$.  Since $E(L)=T\cdot\Delta$ is a sub-system of a prolongation of $\Delta$, this means $X_\al$ is only a sub-symmetry of $\Delta$.

3. Since $\beta=T^*\gamma$ for some smooth $\gamma$, we have $\beta\cdot f=(T^*\gamma)\cdot f=\gamma\cdot(Tf)+div=div$, so we conclude $X_\al L_f=X_\al L+div$.
\end{proof}

\begin{rem}
Part 3 illustrates that a variational symmetry $\al$ of $L$ need not correspond to a symmetry of $\Delta$, since if $X_\al\Delta_f=A\Delta_f$, then $X_\al f=Af$ need not be true for all $f\in ker[T]$.  We will demonstrate this failure in Section \ref{Ex1sec}.

\medskip
Therefore, the Green-Lagrange-Noether approach is incomplete: it requires vector field $X_\al$ to be a symmetry of $\Delta$, while Part 2 indicates we must instead consider sub-symmetries.  The incompleteness was partially demonstrated in \cite{Anco2017} in the case $\beta$ is a characteristic of a conservation law, but the case where cosymmetry $\beta$ is not a characteristic was not indicated.

\medskip
Note that Ibragimov's approach \cite{Ibragimov2011} is distinct; it uses variational symmetries an extended, Lagrangian system obtained by, essentially, treating the $\beta$'s as \textit{dependent variables}.  Such variational symmetries \textit{must} be symmetries of $\Delta$, which is one reason why our result is different.  Another follows from a direct computation of the characteristic; there is an additional term due to the symmetric action of $X_\al$ on $L$, which in our approach is simply a divergence not contributing to the characteristic.
\end{rem}

\medskip
\begin{rem}
The main difference of our extension is that the vector field $X_\al$ is only a sub-symmetry of the PDE $\Delta$, not a true symmetry, since we are considering variational symmetries of $L$.  This highlights the incompleteness of the previous approach, which requires using ordinary symmetries.  We presented an example where all
lower conservation laws arise in this way, many only from sub-symmetries.
\end{rem}

\begin{rem}
Upon restricting said symmetries of $\Delta$ to variational symmetries of $L$, the results are equivalent, of course.  Let $A\Delta=X_\al\Delta$. The Green-Lagrange-Noether approach yields the following conservation law:
\eqal{
(T^*\al -X_\al\beta-A^*\beta)\Delta=div.
}
But in fact, $(X_\al\beta+A^*\beta)=X_\al L+div=div$ since $\al$ is a variational symmetry. Thus, 
\eqal{
(T^*\al-\mu)\Delta=div
}
for some generating function $\mu$.  This implies $T^*\al$ is also a generating function.
\end{rem}

\begin{rem}
We note that $L_0=L|_{\Delta=0}$ is nonzero (mod divergence) in general.  For example, if $E(L)=\Delta$ for $L=uv_t-\lambda[u,v]$ and the two-component evolution system
\eqal{
E_u(L)=v_t-E_u\lambda=0,\\
E_v(L)=-u_t-E_v\lambda=0,
}
then unless $\lambda$ is of a special form, $L$ need not vanish on $\Delta=0$.
\end{rem}

\subsection{An example} 
\label{Ex1sec}
It is well known that the second order Burgers equation
\eqal{
u_t=uu_x+u_{xx}=D_x(\frac{1}{2}u^2+u_x)
}
has exactly one conservation law and no cosymmetries other than $\beta=1$.  It thus has no quasi-Lagrangians, and its conservation law does not arise from such.  We next consider the opposite situation.

\medskip
Consider a third order evolution equation $u_t=F[u]$ of the form
\eqal{\label{Ex1}
u_t=F[u]:=\frac{3u_{xx}u_{xxx}}{u_x}-\frac{u_{xx}^3}{u_x^2}+f(t)=\frac{1}{u_{xx}}D_x(\frac{u_{xx}^3}{u_x}+f(t)u_x).
}
Even if $f(t)=0$, this cannot be written as a Hamiltonian equation, since although $u_t=\ms D\cdot E(H)$ for $H=-\frac{1}{2}u_x^2$ and the skew-adjoint operator
$$
\ms D=2\frac{u_{xx}}{u_x}D_x+\frac{u_{xxx}}{u_x}-\frac{u_{xx}^2}{u_x^2},
$$
it can be shown that $\ms D$ does not verify the Jacobi identity.  We do, however, see a quasi-Lagrangian structure:
\eqal{\label{Ex1quasi}
E(L):=E(-\frac{1}{2}u_x(u_t-F))=D_x(u_t-F)=:T\Delta.
}
We present the cosymmetries of order $\beta(t,x,u,\dots,u_5)$:
\eqal{\label{Ex1adj}
u_x,\quad -u_{xx},\quad -D_x(xu_x+tF),\quad -D_xF.
}

\medskip
The even-order cosymmetries in \eqref{Ex1adj} are conservation law generators which arise from variational symmetries of the quasi-Lagrangian
$$
L:=-\frac{1}{2}u_x(u_t-F[u]).
$$
Indeed, if we rewrite them using \eqref{Ex1},
\eqal{\label{Ex1sym}
-D_xu_x,\quad -D_x(xu_x+tu_t),\quad -D_xu_t,
}
then we see they generate a translation $x\to x+\ep$, a scaling $(x,t)\to e^\ep(x,t)$, and a translation $t\to t+\ep$, respectively.  These (three) cosymmetries comprise the 
(order 4) conservation law generators for equation \eqref{Ex1}.

\medskip
There are two differences from the classical Noether theorem:

1. The variational symmetry $u\to u+\ep$ does not lead to a conservation law since for $\al=1$, we have $T^*(\al)=-D_x(1)=0$.

2. Time translation is \textit{not} a symmetry of $\Delta=0$ unless $f'(t)=0$.  Nevertheless, since 
$$
\frac{1}{2}u_xf(t)
$$ 
is a divergence expression, we see $\al=u_t$ generates a variational symmetry of $L$, and then a conservation law of $\Delta$.


\medskip

\medskip
\begin{rem}
We present a comparison with the quasi-Noether/nonlinear self-adjointness/Green-Lagrange approach.  Suppose the symmetry condition $X_\al \Delta=A\Delta$, and integrate by parts:
\eqal{\label{Xsym}
X_\al L=\beta X_\al\Delta+(X_\al\beta)\Delta=\beta A\Delta+(X_\al\beta)\Delta=(A^*\beta+X_\al\beta)\Delta+div.
}
Also recall the Noether identity and quasi-Lagrangian structure \eqref{Ex1quasi}
\eqal{\label{XNoe}
X_\al L=\al E(L)+div=\al D_x \Delta+div=(-D_x\al)\Delta+div.
}
Combining these two formulas gives the characteristic for the conservation law generated using this approach:
\eqal{
B(\al):=-D_x\al-X_\al\beta-A^*\beta.
}
The only symmetry in \eqref{Ex1sym} is $\al=u_x$.  In this case, it can be verified that
$$
-X_\al\beta=\frac{1}{2}D_x\al,\quad X_\al\Delta=D_x\Delta,\quad -A^*\beta=-\frac{1}{2}D_x\al,
$$
so that $B(\al)=-D_x\al=T^*(\al)$.  Recall the characteristic generated by a variational symmetry is also $-D_x\al$.  Thus, the two approaches agree when $\al$ is a symmetry of \eqref{Ex1}.
\end{rem}

\begin{rem}
Observe that $\al=u_t$ and $\al=tu_t+xu_x$ do not generate symmetries of \eqref{Ex1} if $f'(t)\neq 0$.  However, a simple modification yields the desired conservation law.  

Combining \eqref{Xsym} with \eqref{XNoe} without the symmetry condition $X\Delta\dot= 0$ yields
\eqal{
(-D_x\al-X_\al\beta)\Delta-\beta\,X_\al\Delta=div.
}
For $\al=u_t$, it can be shown that
\eqal{
X_\al\Delta=D_t\Delta-f'(t)=:R\Delta-f'(t).
}
Therefore, after an integration by parts, we recover the conservation law plus an error term:
\eqal{
(-D_x\al-X_\al\beta-A^*\beta)\Delta+\beta f'(t)=div.
}
In fact, the error is a divergence: $f'(t)\beta=D_x(-\frac{1}{2}uf'(t))$, so we derive an analogous conservation law.  Again, the characteristic is $-D_x\al$.  

\medskip
Let us apply the same analysis to $\al=tu_t+xu_x$.  We have
$$
X_\al\Delta=-\Delta+tD_t\Delta+xD_x\Delta+(1-tD_t)f(t)=:A\Delta+f(t)-tf'(t).
$$
As before, we obtain
$$
(-D_x\al-X_\al\beta-A^*\beta)\Delta+\beta(f(t)-tf'(t))=div.
$$
The second expression on the left hand side is a divergence, so we obtain a conservation law.  Using $-A^*\beta=-\frac{1}{2}D_x\al$, we conclude the characteristic is again $-D_x\al$.
\end{rem}

\subsection{Critical points and symmetries}
We will show that quasi-Lagrangians are ``opposite" in a sense to Hamiltonians.  Let $\Delta=u_t-P[u]$ be an evolution system with evolutionary operator
\eqal{
D_t=\pd_t+u^a_{It}\pd_{u^a_I}\deq\pd_t+D_IP^a\pd_{u^a_I}=\pd_t+X_P.
}
Each of these objects satisfies special determining equations: $\al$ symmetry, $\beta$ cosymmetry.  Symmetry $\al$ satisfies the equation
\eqal{
D_t\al=\mbb D_P\al.
}
Cosymmetry $\beta$ satisfies the equation
\eqal{
D_t\beta=-\mbb D_P^*\beta.
}
These equations imply that \textit{symmetry invariance} $\al=0$ and \textit{cosymmetry invariance} $\beta=0$ are compatible with $\Delta=0$.  By \textit{compatible}, we mean if $u_0(x)$ solves $\gamma(t_0,x,u_0,\dots)=0$, then for time evolution $u(t)=U(t;t_0)u_0$, we have $\gamma(t,x,u(t),\dots)=0$.  This follows for $P,u_0$ analytic by a (local) power series expansion, since the invariance conditions imply $D_t^k\gamma|_{\gamma=0}=0$ for all $k\ge 0$.  The same holds for systems in Cauchy-Kovalevskaya form, which can be rewritten as evolution systems.

\begin{rem}
It was observed in \cite{Anco2017} that a co-symmetry $\beta$ induces an invariant one-form $\beta[u]\cdot du$ with respect to time evolution, in an analogous way that a symmetry $\al$ induces an invariant vector field $\alpha[u]\cdot\pd/\pd u$.  The determining equations are simply vanishing Lie derivatives of these tensors.  The critical sets of these tensors are therefore time invariant, which is a more geometric way to describe the compatibility of these invariance conditions.
\end{rem}

Let us note an interesting phenomenon. If $E(L)=T\Delta$ is quasi-Lagrangian and $X_\al$ is a variational symmetry, then $X_\al$ is a symmetry of $T\cdot\Delta$, but $T^*\al$ generates a conservation law of $\Delta$.  We thus find a ``duality" relation between the existence of compatible systems.

\begin{prop}
If $E(L)=T\Delta$ and $X_\al L=div$, then both $\{T^*\al=0,\Delta=0\}$ and $\{\al=0,T\Delta=0\}$ are compatible systems.
\end{prop}

\medskip
Next, if cosymmetry $\beta$ also generates a conservation law, we call the equation $\beta=0$ the \textit{critical point condition} for the conserved integral generated by $\beta$.  To justify this terminology, we recall the well known fact that cosymmetry $\beta$ is a characteristic if and only if $\mbb D_\beta^*-\mbb D_\beta=0$ (self-adjointness), i.e. $\beta=E(T)$ for the conserved density $T$ of the conservation law $D_tT+D_iX^i=\beta\cdot\Delta=0$.  It follows that $\beta=0$ is the Euler-Lagrange equation for the (possibly time-dependent) functional $\int T[u]dx$, i.e. the equation for critical points of this conserved integral.  

\medskip
By the discussion for cosymmetries, the critical point condition is compatible with time evolution $u_t=P$.  This is sensible for two reasons: 1. if $u(0)$ is a minimizer of $\int T[u]dx$ for suitable boundary conditions, then the conservation of $\int T$ implies $u(t)$ is also a minimizer at later times.  2. If $\pd_t T=0$, then 
\eqal{
X_PT=D_tT=-D_iX^i=div,
}
so $P$ actually generates a variational symmetry of the functional $\int Tdx$.  It follows that $X_P$ is a symmetry of $E(T)=0$, hence preserves the solution space.

\medskip
\begin{rem}
Although many current papers are devoted to the construction of Lie-type invariant solutions and conservation laws for non-Hamiltonian systems, we are not aware of any which constructed the critical points of these conservation laws.  The possibility for such was raised in \cite{Ma}.
\end{rem}


\begin{ex}[Time-dependent conservation law]
The KdV equation
\eqal{
u_t+uu_x+u_{xxx}=0,
}
has conservation law 
\eqal{
D_tT+D_xX\deq 0
} 
of the form
\eqal{
\pdf{t}\left(xu-\frac{tu^2}{2}\right)+\pdf{x}\left(t\left(\frac{u_x^2}{2}-uu_{xx}-\frac{u^3}{3}\right)+\frac{xu^2}{2}+xu_{xx}-u_x\right)=0,
}
such that $T[u]=xu-\frac{tu^2}{2}$, and $E_1T[u]=x-tu$.  If $\sigma\neq0$, then $u_*(x)=x/\sigma$ is a solution of $E_1T(\sigma,x,u,\dots)=0$, and $u(t,x)=x/t$ is the solution of $\{u_t=P,E_1T=0\}$ which satisfies $u(\sigma,x)=u_*(x)$.
\end{ex}

\begin{ex}[Time-dependent evolution]
The generalized KdV equation \cite{AncoMaria}
\eqal{\label{tkdv}
u_t+f(t,u)u_x+u_{xxx}=0,
}
where $f(t,u)=at^{-1/3}u+bu+cu^2$ for $a,b,c$ constant, has explicit time dependence if $a\neq 0$.  It has a conservation law $D_tT+D_xX\deq0$ with conserved density
\eqal{
T=\frac{1}{2}ctu_x^2-\frac{1}{12}t(cu^2+bu)^2+\frac{1}{6}x(cu^2+bu)-\frac{1}{2}at^{2/3}(\frac{1}{3}cu^3+\frac{1}{4}bu^2),
}
and characteristic (note that \cite{AncoMaria} has a typo in the $u_{xx}$ term)
\eqal{
E_uT=-ctu_{xx}-\frac{1}{6}t(2c^2u^3+3bcu^2+b^2u)-\frac{1}{4}at^{2/3}(2cu^2+bu)+\frac{1}{6}x(2cu+b).
}
Let show the system $\{u_t+fu_x+u_{xxx}=0,E_uT=0\}$ is compatible.  Suppose first that $c=0$.  Then the solution of $E_uT=0$ is
\eqal{
u(t,x)=\frac{2x}{(3a+2bt^{1/3})t^{2/3}},
}
which solves \eqref{tkdv} for $c=0$.  Suppose now that $c\neq 0$.  If we solve $E_uT=0$ for $u_{xx}$, then it is easy to show that $(u_{xx})_t=(u_t)_{xx}$, or that \eqref{tkdv} is consistent with $E_uT=0$.  Alternatively, substituting $E_tT=0$ into \eqref{tkdv} gives a first order PDE, 
which has solution
\eqal{
u(t,x)=-\frac{b}{2c}+t^{-1/3}g\bigl(\xi\bigr),\qquad \xi=t^{-1/3}x+(4c)^{-1}(3abt^{1/3}+b^2t^{2/3}),
}
where $g$ is an arbitrary function.  Substitution into $E_uT=0$ gives an ODE for $g$.
\end{ex}

\begin{ex}[A non-Hamiltonian example]
The nonlinear telegraph system \cite{Bluman2005}
\eqal{\label{telegraph}
v_t+ke^uu_x-e^u=0,\\
u_t-v_x=0,
}
where $k\neq 0$ is a constant, has a conservation law with density (note the typo in \cite{Bluman2005})
\eqal{
T=e^{-x/k}[(tv/2+2x)v-k(uv+te^u)]
}
and characteristic
\eqal{
E_uT&=-ke^{-x/k}(v+te^u),\\
E_vT&=e^{-x/k}(2x+tv-ku).
}
A critical point $(u_*,v_*)$ satisfies the system $E_uT=E_vT=0$, whose solution is
\eqal{
u_*(x,t)=2x/k-W\bigl(\frac{t^2}{k}e^{2x/k}\bigr),\\
v_*(x,t)=-\frac{k}{t}W\bigl(\frac{t^2}{k}e^{2x/k}\bigr),
}
where $W$ is the Lambert W function.  It is straightforward to show that $(u_*,v_*)$ solves \eqref{telegraph}.
\end{ex}

\medskip
Before turning to our next result, we recall Hamiltonian systems, see Chapter 7 in \cite{Olver}.  Suppose 
\eqal{
P=\ms D\cdot E(H)
} 
for some Hamiltonian $H[u]$ and skew-symmetric operator $\ms D$ which verifies the Jacobi identity, in the sense that if $\{\ms P,\ms Q\}=\int E(P)\cdot\ms D\cdot E(Q)dx$ is the Poisson bracket induced by $\ms D$, then $\{,\}$ verifies the Jacobi identity.  Then the Noether relation (Theorem 7.15 in \cite{Olver}) states that every conserved integral $\int Tdx$ yields a symmetry of Hamiltonian form: $\al=\ms D\cdot E(T)$.  In other words, 
\eqal{
\al=\ms D\cdot\beta,
} 
where $\beta=E(T)$ is the associated characteristic.  The converse is not true: if the Hamiltonian vector field generated by $\al=\ms D\cdot E(T)$ is a symmetry, then $\beta=E(T)$ need not be a characteristic.  However, there exists a time-dependent $C[t;u]$ with $C[t_0;.]\in ker\ms D$ for each $t_0$, such that $\beta-C$ generates a conservation law.

\medskip
We now compare this situation to the ``opposite" quasi-Lagrangian case in the following.

\begin{thm}
Let $E(L)=T\cdot\Delta$. If $\gamma=T^*\al$ generates a conservation law, then vector field invariance implies the critical point condition.  If $T$ is a smoothly invertible (skew-symmetric) matrix function, then the conditions are equivalent.

Let $u_t=\ms D\cdot E(H)$.  If $\al=\ms D\cdot \beta$ for $\beta=E(G)$, then the critical point condition implies symmetry invariance.  If $\ms D$ is an invertible (skew-symmetric) matrix function, then the conditions are equivalent.

If $T,\ms D$ commute with $\pd_t$, then their invertibility implies that $T^{-1}\Delta,\ms D^{-1}\Delta$ is Lagrangian, respectively.
\end{thm}

\begin{rem} In the first case, $\al$ need not be a symmetry, but if it is, then symmetry invariance implies the critical point condition.  Moreover, in Section \ref{Ex1sec}, we see that symmetry invariance occurs after dividing by the kernel $f(t)$ of $T$, analogous to conservation of $\beta$ holding after dividing by the kernel of $\ms D$.
\end{rem}

\section{Conclusions}
In this paper, we discussed the problem of correspondence between symmetries and conservation laws for a general class of  differential systems (quasi-Noether systems). Our approach is based on the Noether operator identity. We discussed some properties of Noether identity, and showed that it leads to the same conservation laws as Green-Lagrange identity. We generated classes of quasi-Noether equations of the second order of evolutionary and quasilinear form. 

\medskip
We 
introduced the notion of a \textit{quasi-Lagrangian}, which generalizes the quantity $L:=\beta\cdot\Delta$ used in previous works to the case where $L|_{\Delta=0}$ may not be zero.  This generalization recovers the case where $\Delta=E(L)$ is actually a Lagrangian system, thus allowing us to extend Noether theorem 
(Theorem \ref{thm:qL}).  The previous work was not able to achieve both of the following: 1. achieve an extension of Noether while still involving symmetry vector fields (see \cite{Anco2017} for an extension using only cosymmetries), and 2. recover the case of a Lagrangian system $\Delta=E(L)$.

Based on the notion of invariant submanifolds we introduced critical points of conservation laws, and gave examples of the compatibility of cosymmetry invariance and the critical point condition, including for a non-Hamiltonian system.

\medskip
We concluded with a comparison of the invariant submanifolds of quasi-Noether and Hamiltonian evolution equations, and showed these systems are ``opposite" in some sense.  For Lagrangian systems, the cosymmetry and symmetry invariant submanifolds coincide, while in general there is a containment in one direction or the other.

\end{document}